\documentclass{article}

\usepackage{arxiv}

\usepackage[utf8]{inputenc} 
\usepackage[T1]{fontenc}    
\usepackage{hyperref}       
\usepackage{url}            
\usepackage{booktabs}       
\usepackage{amsfonts}       
\usepackage{nicefrac}       
\usepackage{microtype}      
\usepackage{lipsum}		
\usepackage{graphicx}
\usepackage{natbib}
\usepackage{doi}

\usepackage{color}
\usepackage{amsmath,amssymb,amsthm}
\newtheorem{lemma}{Lemma}
\newtheorem*{pf}{Proof}
\newtheorem{assum}{Assumption}
\newtheorem{theorem}{Theorem}

\newtheorem{remark}{Remark}
\newtheorem{problem}{Problem}
\newtheorem{subproblem}{Subproblem}
\newtheorem{definition}{Definition}

\title{Control Barrier Function only Formation Tracking in Multi-Agent Systems}

\author{ \href{https://orcid.org/0000-0003-0212-0762}{\includegraphics[scale=0.06]{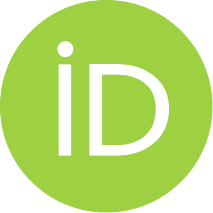}\hspace{1mm}S. Saharsh} \\
	Department of Cyber Physical Systems\\
	Indian Institute of Science\\
	Bengaluru, 560012, India \\
	\texttt{saharsh2021@iisc.ac.in} \\
	\And
	Pushpak Jagtap \\
	Department of Cyber Physical Systems\\
	Indian Institute of Science\\
	Bengaluru, 560012, India\\
	\texttt{pushpak@iisc.ac.in} \\
}

\renewcommand{\shorttitle}{\textit{arXiv} Template}

\begin{document}
\maketitle

\begin{abstract}
This paper presents a real-time control framework for formation tracking of heterogeneous multi-agent systems with non-linear dynamics. The proposed method formulates a \emph{single} Control Barrier Function-like constraint within a quadratic optimization setting that addresses formation tracking. Relying on the relative information of neighboring agents, the controller is designed to operate without the need for manual parameter tuning or a separate nominal formation controller. The leader-follower framework is validated through simulations of moving formations.
\end{abstract}


\section{Introduction}
A multi-agent system (MAS) comprises autonomous agents that coordinate to accomplish shared objectives, enhancing reliability and robustness compared to a single agent. Primary research focuses on consensus, formation control, circumnavigation, and collision avoidance for applications such as security, surveillance, and exploration.
The distributed formation control problem (\cite{liu2024survey}) in multi-agent systems has been tackled using position, distance, bearing, or displacement-based information from connected agents. Real-time communication constraints prevent agents from accessing global state information, allowing only relative state measurement with nearby neighbors. This motivates the use of relative state measurements available in the local vicinity of each agent, obtained using a camera, LiDAR (\cite{tron2016distributed}) or wireless sensors (\cite{mao2007wireless}).

Lyapunov-based formation control strategies (\cite{rayabagi2024formation}, \cite{pampatwar2021planar}), leveraging bearing Laplacian matrix (\cite{zhao2015bearing}), address the problem of sensing limitations by relying solely on bearing vectors, partial information that reduces sensor dependency and is suitable for Global Positioning System (GPS)-denied environments, but suffers from either poor scalability because of its reliance on specific neighbor relationships or lacks adaptability in highly dynamic or uncertain environments. In contrast, this paper investigates the moving formation tracking problem in a multi-agent system using relative state information of neighboring agents.
Recent advances (\cite{huang2024prescribed}, \cite{fu2023safe}) in the formation control of MAS have shown that a Quadratic Programming (QP)-based distributed control strategy for formation tracking enables real-time constraint handling, scales efficiently with the number of agents, integrates nominal feedback controllers for formation, and ensures optimal control effort. The nominal formation controllers in the QP problem typically conflict with real-time constraints, causing infeasibility (also highlighted in \cite{wu2023quadratic}). 

This work makes three main contributions that address the key limitations of existing approaches. $(i)$ We propose a CBF-like function that extends standard CBFs by ensuring both forward invariance of a safe set and finite-time reachability of a target set within the safe set, whereas conventional CBFs guaranty only forward invariance of the safe set. $(ii)$  We introduce a novel CBF-QP-like optimization framework with \emph{single} CBF constraint for formation tracking with a heterogeneous nonlinear multi-agent system, $(iii)$ We achieve formation tracking with `non-linear agent dynamics' for multi-agent systems in the leader-follower framework, using local bearings, relative positions, and velocities of neighboring agents.

\section{Problem Formulation}
\label{sec: prelim and prob defn}

\textit{Notations}. Scalars are denoted as $x \in \mathbb{R}$, and vectors as $\boldsymbol{x}$ (bold letter).  Let $\mathcal{S}$ be the unit circle.  Given a set $\mathcal{C}$, $\partial \mathcal{C}$ denotes its boundary and $int (\mathcal{C})$ is the interior. The cardinality of the set $A$ is denoted by $|A|$. 
The closed ball $\mathbb{B}(\boldsymbol{c},\lambda) \subseteq \mathbb{R}^2$ is defined as $\mathbb{B}(\boldsymbol{c}, \lambda) := \{ \boldsymbol{x} \in \mathbb{R}^2 : \|\boldsymbol{x} - \boldsymbol{c}\| \leq \lambda\}$, where $\boldsymbol{c}$ is the center and $\lambda>0$ is the radius. All other notation follows standard mathematical conventions.

\subsection{System Description}
We consider a heterogeneous multi-agent system consisting of $n$ agents, where each agent is indexed by $i \in \{1, 2, \dots, n\}$. Each $i^{th}$ agent (modeled as a planar robot) is represented as:
\vspace{-0.5em}
\begin{align}
        \dot{\boldsymbol{p}}_i &= \boldsymbol{v}_i,  \nonumber \\
        \dot{\boldsymbol{v}}_i &= f_{v_i}(\boldsymbol{p}_i,\boldsymbol{v}_i) + h_{v_i}(\boldsymbol{p}_i,\boldsymbol{v}_i)\boldsymbol{u}_i, \label{sys}
\end{align}
    where $\boldsymbol{p}_i,\boldsymbol{v}_i\in \mathbb{R}^2$ and $\boldsymbol{u}_i\in \mathbb{R}^2$ are the position in the 2D plane, velocity, and input vectors, respectively. The maps $f_{v_i} : \mathbb{R}^2 \times \mathbb{R}^2 \rightarrow \mathbb{R}^2$, $h_{v_i}: \mathbb{R}^2 \times \mathbb{R}^2 \rightarrow \mathbb{R}^{2 \times 2}$ are locally Lipschitz continuous functions.

\begin{remark}
    A broad class of planar robotic systems can be modeled as \eqref{sys}, such as simple double integrator, Ackermann-drive vehicles, differential-drive robots, aerial vehicles and fixed-wing aircraft operating in altitude hold mode. {For detailed derivation, especially for a system with orientation state, refer to \cite{tayal2026collision} and \cite{sawarkar2026sliding}[II B].}
\end{remark}

\subsection{Graph-based Communication Framework}
Given a multi-agent system with $n$ agents, the communication topology is represented by an undirected graph $\mathcal{G} = (\mathcal{V},\mathcal{E})$, where the set of nodes of the graph $\mathcal{V} = \{1,2,...,n\}$ denotes the set of $n$ agents (modeled as \eqref{sys}) and $\mathcal{E} \subseteq \mathcal{V} \times \mathcal{V}$ denotes the set of edges, with each edge given by a pair $(i,j) \in \mathcal{E}$, which indicates the flow of information from the $j^{th}$ agent to the $i^{th}$ agent. The neighbors of the $i^{th}$ agent are denoted as $\mathcal{N}_{i}$, where $\mathcal{N}_{i} := \{j \in \mathcal{V}: (i,j) \in \mathcal{E}, j \neq i \}$.\\ 
For a given edge $(i,j) \in \mathcal{E}$, we denote the relative position between agent $i$ and $j$ by $\boldsymbol{p}_{ij} := \boldsymbol{p}_{j} - \boldsymbol{p}_{i}$ and the bearing by $\boldsymbol{g}_{ij} := \dfrac{\boldsymbol{p}_{ij}}{\|\boldsymbol{p}_{ij}\|}, \forall (i,j) \in \mathcal{E}$.
The bearing configuration $\boldsymbol{g} \in \mathbb{R}^{2|\mathcal{E}|}$ is formed by stacking the bearing vectors $\boldsymbol{g}_{ij}$ for all edges $(i,j) \in \mathcal{E}$.
Given $\boldsymbol{g}_{ij} \in \mathbb{R}^2$, orthogonal projection matrix $P_{\boldsymbol{g}_{ij}} \in \mathbb{R}^{2 \times 2}$ is defined as $P_{\boldsymbol{g}_{ij}} = I_2 \ - \ \boldsymbol{g}_{ij} \boldsymbol{g}_{ij}^{\top}$.
Using the matrix $P_{\boldsymbol{g}_{ij}}$ and the undirected graph $\mathcal{G} = (\mathcal{V},\mathcal{E})$, the bearing Laplacian matrix $\mathcal{B} \in \mathbb{R}^{2n \times 2n}, n = |\mathcal{V}|$ (as in \cite{zhao2019bearing}), is defined as:
\begin{align}\label{bear}
    [\mathcal{B}]_{ij} =
\begin{cases}
{0}_{2}, & i \ne j,\ (i,j) \notin \mathcal{E}, \\[8pt]
- P_{\boldsymbol{g}_{ij}}, & i \ne j,\ (i,j) \in \mathcal{E}, \\[8pt]
\sum\limits_{k \in \mathcal{N}_i} P_{\boldsymbol{g}_{ik}}, & i = j,\ i \in \mathcal{V}.
\end{cases}\end{align}
The matrix $\mathcal{B}$ is used later in Subsection \ref{subsec:problem formulation} to define the bearing rigidity. 


In this work, we consider a leader-follower framework with the objective of formation tracking. In this setup, agents are divided into two categories: leaders, denoted by $ \mathcal{V}_{L} \subset \mathcal{V} $ with $ n_l = |\mathcal{V}_{L}| $, and followers, denoted by $\mathcal{V}_{F} = \mathcal{V} \setminus \mathcal{V}_{L} $ with $ n_f = |\mathcal{V}_{F}| = n - n_l $. The leader agents follow a predetermined path (such as a straight path), which serves as the anchor for the moving formation. Follower agents are required to maintain a prescribed formation, guided by the desired bearing vectors. Bearing Laplacian matrix $\mathcal{B}$ (as defined in \eqref{bear}) for the leader-follower framework can be partitioned into block matrices as 
    $\mathcal{B} = \begin{bmatrix}
\mathcal{B}_{ll} & \mathcal{B}_{lf} \\
\mathcal{B}_{fl} & \mathcal{B}_{ff}
\end{bmatrix}$,
where $\mathcal{B}_{ll} \in \mathbb{R}^{2n_l \times 2n_l}$, $\mathcal{B}_{lf} \in \mathbb{R}^{2n_l \times 2n_f}$,
$\mathcal{B}_{fl} \in \mathbb{R}^{2n_f \times 2n_l}$, $ \mathcal{B}_{ff} \in \mathbb{R}^{2n_f \times 2n_f}$, where $l$ and $f$ represent leaders and followers, respectively.
The formation objective with inter-neighbor bearings is formally defined in the following subsection.

\subsection{Problem Formulation}
\label{subsec:problem formulation}
The multi-agent system of $n$ agents, connected over a communication graph $\mathcal{G}$ with the leader-follower framework, is required to reach, maintain, and track the desired formation, given by the desired leader positions $\boldsymbol{p}^{d}_{l}, \forall l \in \mathcal{V}_{L}$ and the desired bearing vectors $\boldsymbol{g}_{ij}^{d}$. The unique formation of the multi-agent system is characterized by a desired position configuration vector $\boldsymbol{p}^{d} \in \mathbb{R}^{2n}$, defined as
\begin{align}\label{dpos}
   \boldsymbol{p}^{d} := [(\boldsymbol{p}_{1}^{d})^{\top}, \cdots,(\boldsymbol{p}_{n}^{d})^{\top}]^{\top}, 
\end{align}
where $\boldsymbol{p}_{i}^{d}$, $i\in\mathcal{V}$ denote the desired formation position of agent $i$. 
The formation objective is defined below.
\begin{definition} [Unique formation]
    Given a desired position configuration $\boldsymbol{p}^{d}$, with agents (modeled as \eqref{sys}) in the multi-agent system, connected as $\mathcal{G}$ under a leader-follower framework, it is said to achieve a \textit{unique formation} if the following holds: $\lim_{t \rightarrow \infty} \| \boldsymbol{p}_i(t) - \boldsymbol{p}_i^{d}\| = 0, \forall i \in \mathcal{V}.$
    \label{uform}
\end{definition}
With the given leader positions $\boldsymbol{p}_l^d, l \in \mathcal{V}_L$, the desired formation in the bearing configuration $\boldsymbol{g}^d$, denoted as $\mathcal{F}(\boldsymbol{g})$, described by the desired bearing vectors $\boldsymbol{g}_{ij}^d := \frac{\boldsymbol{p}_{ij}^{d}}{\|\boldsymbol{p}_{ij}^{d}\|}, \forall (i,j) \in \mathcal{E}$, $\boldsymbol{p}_{ij}^{d}= \boldsymbol{p}_{j}^{d} - \boldsymbol{p}_{i}^{d}$, is achieved if
    \begin{align}\label{form}
        &\lim_{t \rightarrow \infty} \langle\boldsymbol{g}_{ij}(t),\boldsymbol{g}_{ij}^{d}\rangle = 1, \forall (i,j) \in \mathcal{E}. 
    \end{align}
We assume bearing rigidity to ensure that a bearing-defined formation corresponds to a unique formation as in Definition \ref{uform}, which is as follows.
\begin{assum}[\cite{zhao2016localizability}]
     The graph $\mathcal{G}$ is infinitesimally bearing rigid, i.e., $det(\mathcal{B}_{ff}) \neq 0$. 
    \label{A1}
\end{assum}
The transition from a formation $\mathcal{F}(\boldsymbol{g})$ defined in bearing to the unique formation is established next.
\begin{lemma}[\cite{zhao2016localizability}]
    Given a multi-agent system in the leader-follower framework, connected as $\mathcal{G}$ with leaders positions as $\boldsymbol{p}_{l}^{d}$ and in distinct positions, i.e., $\boldsymbol{p}_{l_1} \neq \boldsymbol{p}_{l_2}, \forall l_1,l_2 \in \mathcal{V}_L$, the desired formation is uniquely determined with $\mathcal{F}(\boldsymbol{g})$, i.e., $\forall (i,j) \in \mathcal{E}, \lim_{t \rightarrow \infty} \boldsymbol{g}_{ij}(t) = \boldsymbol{g}_{ij}^{d} \Rightarrow \forall i \in \mathcal{V}_{F}, \lim_{t \rightarrow \infty} \boldsymbol{p}_{i}(t) = \boldsymbol{p}_{i}^{d}$, if and only if $\mathcal{B}_{ff}$ is non-singular. In addition, if $\mathcal{B}_{ff}$ is non-singular, then the system has at least two leaders (i.e., $n_l \geq 2$). 
    \label{lem1}
\end{lemma}
\begin{definition} [Rigid formation]
   Consider a multi-agent system connected as $\mathcal{G}$ under a leader-follower framework, with time-varying leader trajectories $\boldsymbol{p}_l^{d}(t)$ moving with constant velocity $\boldsymbol{v}_{l}$ and the desired bearing configuration $\boldsymbol{g}^{d}$. With Assumption \ref{A1}, the agents are said to achieve a unique \textit{rigid formation} $\mathcal{F}(\boldsymbol{g}(t))$ if $
    \lim_{t \rightarrow \infty} \| \boldsymbol{g}(t)- \boldsymbol{g}^{d}\| = 0$.  
\label{mform}
\end{definition}
\begin{remark}
    All leader agents $i \in \mathcal{V}_L$ have constant velocity $\boldsymbol{v}_{l}$ to maintain the unique formation shape under $\mathcal{F}(\boldsymbol{g})$. 
\end{remark}
Additionally, the desired position vectors $\boldsymbol{p}_{i}^d$ also move with velocity $\dot{\boldsymbol{p}}_{i}^d = \boldsymbol{v}_{l}$. The problem addressed in this work is defined as follows.
\begin{problem}\label{mainprob}
    Given a heterogeneous multi-agent system with agent dynamics as in \eqref{sys}, a communication graph $\mathcal{G}$ with leader-follower framework under Assumption \ref{A1}, a desired unique rigid formation $\mathcal{F}(\boldsymbol{g}(t))$ specified by the leader's trajectory $\boldsymbol{p}_{l}^{d}(t)$ and desired bearing vectors $\boldsymbol{g}_{ij}^{d}$, design a distributed control strategy using bearing information $\boldsymbol{g}_{ij}$, so that the controlled system achieves a unique rigid formation $\mathcal{F}(\boldsymbol{g}(t))$ as in Definition \ref{mform}.
\end{problem}
Consequently, we propose a novel time-varying formation tracking control law to solve Problem \ref{mainprob}, based on a single CBF like constraint in the QP optimization framework. 
{Unlike standard CLF-CBF formulations where formation and safety pose as competing objectives, the proposed approach encodes formation itself as a CBF constraint, enabling a direction for unified constraint-based safe control design (using \cite{wang2017safety}). CBF-like constraint strategy is explained next.

\subsection{Reachability Control Barrier Function}
We introduce a CBF like function that ensures forward invariance of a safe set $X \subset \mathbb{R}^{n}$ as well as reachability of a target set $R \subset X$, in finite time $\tau \in \mathbb{R}^{+}$. We employ a continuously differentiable function $b: \mathbb{R}^{n} \rightarrow \mathbb{R}$ in order to define the safe set $X$, where
\begin{align} \label{safeset}
    X:= \{{\boldsymbol{x}} \in \mathbb{R}^{n}: b(\boldsymbol{x}) \geq 0\},
\end{align}
and the boundary set $\partial X := \{{\boldsymbol{x}} \in \mathbb{R}^{n}: b(\boldsymbol{x}) = 0 \}$. Using the function $b$, which has unique maxima in $X$, $\max_{\forall \boldsymbol{x} \in X} b(\boldsymbol{x}) = M$, with constant $\mu \leq M, \mu\in \mathbb{R}^{+}_{0}$, the target set $R$ is defined as 
\begin{align} \label{targetset}
    R:= \{{\boldsymbol{x}} \in \mathbb{R}^{n}: b(\boldsymbol{x}) \geq \mu\}.
\end{align}
Now we define our CBF like function, using $b$ as below.
\begin{definition}[Reachability CBF]\label{cbfdef}
    Given a compact set $X \subset \mathbb{R}^n$ (as in \eqref{safeset}), target set $R$ (as in \eqref{targetset}), using a continuously differentiable function $b$, with unique maxima in $X$, associated with the affine control system $\dot{\boldsymbol{x}} = f({\boldsymbol{x}}) + h({\boldsymbol{x}})\boldsymbol{u}$, where $f: \mathbb{R}^{n} \rightarrow \mathbb{R}^{n}$, $h: \mathbb{R}^{n} \rightarrow \mathbb{R}^{n \times m}$,  $\boldsymbol{u} \in \mathbb{R}^{m}$, then $b$ is a \emph{reachability control barrier function} (reach CBF) if for reachability time $ \tau > 0$, it holds that
    \begin{align}
        \label{rcbfcons}
        &\sup_{{\boldsymbol{u}}}\{\underbrace{\nabla b(\boldsymbol{x}) \cdot f(\boldsymbol{x})}_{K_f(\boldsymbol{x})} + \underbrace{\nabla b(\boldsymbol{x}) \cdot h(\boldsymbol{x})}_{K_h^{\top}(\boldsymbol{x})}{\boldsymbol{u}}\} \geq \frac{\mu}{\tau},\forall \boldsymbol{x} \in X,
    \end{align}
where $K_f : \mathbb{R}^n \rightarrow \mathbb{R}$ and $K_h : \mathbb{R}^n \rightarrow \mathbb{R}^{m}$ are locally Lipschitz continuous functions. 
\end{definition}
Furthermore, the CBF-QP like optimization problem is designed for the controller framework, using reach CBF $b$ and the linear constraint \eqref{rcbfcons} in ${\boldsymbol{u}}$ as below.
\begin{theorem}\label{thm1}
Given a compact set $X$ as defined in \eqref{safeset}, a target set $R$ (as in \eqref{targetset}), using the reach CBF $b$, as in Definition \ref{cbfdef}, if $\frac{\partial b}{\partial {\boldsymbol{x}}} \neq 0, \forall {\boldsymbol{x}} \in \partial X$, for initial state $\boldsymbol{x}(0) \in X$, and the control input ${\boldsymbol{u}}$ is computed based on the optimization problem as
\begin{align}\label{QPdef}
    &{\boldsymbol{u}}({\boldsymbol{x}}) = \arg \min_{\boldsymbol{q} \in \mathbb{R}^{m}} \quad \frac{1}{2}\boldsymbol{q}^{\top}\boldsymbol{q} , \nonumber \\
    &s.t. \quad K_{f}(\boldsymbol{x}) + K_{h}^{\top}(\boldsymbol{x})\boldsymbol{q} \geq \frac{\mu}{\tau}, 
\end{align}
where $\tau \in \mathbb{R}^{+}, 0 \leq \mu \leq M$, then the set $R$ is said to be reachable in finite time $\tau$, inside the forward invariant set $X$ with minimal effort.    
\end{theorem}
\begin{pf}
    By using the initial condition $b(\boldsymbol{x}(0)) \geq 0$, and after integrating the constraint in \eqref{QPdef}, we get
\begin{align}
    b(\boldsymbol{x}(\tau)) - b(\boldsymbol{x}(0)) &\geq \int_{0}^{\tau}\frac{\mu}{\tau}ds  = \mu> 0.
    \label{step1r}
\end{align}
Thus, the set $R$ is reached in time $\tau$ as shown in \eqref{step1r}. The compact set $X$ is forward invariant using (\cite{ames2019control}[Theorem 1]), as $b(\boldsymbol{x}(0)) \geq 0, \dot{b} \geq 0, \forall t > 0$. This completes the proof. 
\end{pf}
{In contrast to finite-time or prescribed-time CBFs (\cite{abel2023prescribed}), which primarily enforce safety over a specified time horizon, and CLF-based methods, which focus on asymptotic convergence to a target set, the proposed reach-CBF formulation simultaneously guarantees forward invariance of the safe set while ensuring convergence to the target set within a prescribed time.}
Moving forward, our main Problem \ref{mainprob} is decomposed into subproblems for each follower agent $i$ for the distributed control strategy as discussed in the following section.

\section{Subproblem Formulation for Distributed Execution}
\label{sec: distributed execution}
To solve the Problem~\ref{mainprob} in a distributed manner with the QP optimization framework, we decompose the global formation task $\mathcal{F}(\boldsymbol{g}(t))$, defined over all edges $(i,j) \in \mathcal{E}$, into local subtasks $\mathcal{F}_i(\boldsymbol{g}(t))$, defined over each agent's neighborhood $\mathcal{N}_i$. With the given rigidity using Lemma \ref{lem1}, the problem decomposition enables each agent to pursue a local formation constraint while collectively ensuring convergence to the desired unique formation in a distributed optimization setting. To address Problem~\ref{mainprob} based on local environment information, we assume:
\begin{assum}
    The relative position of $j^{th}$ agent with respect to $i^{th}$ agent can be measured under two cases:
    \begin{itemize}
        \item If they are connected, i.e., if $j \in \mathcal{N}_i$,
        \item If they lie within the local vicinity region $\mathbb{B}(\boldsymbol{p}_i, \lambda_i)$, we define the set $E_i$ of all agents in this vicinity as 
        \begin{align}\label{locvic}
            E_i := \{j \in \mathcal{V}: \boldsymbol{p}_{j} \in \mathbb{B}(\boldsymbol{p}_i, \lambda_i) \},
        \end{align}where $\lambda_i \in \mathbb{R}^+$ is the sensing radius of the $i^{th}$ agent.
    \end{itemize}
    \label{A2}
\end{assum}
\begin{remark}
Assumption \ref{A2} is non-conservative, as it requires only that agent $i$, when connected to or in close proximity with its neighbors, can measure their relative positions. This assumption does not depend on the absolute positions of neighbors; instead, it relies solely on relative measurements that can be obtained using onboard sensors such as cameras or LiDAR. The proposed approach is effective even in GPS-denied environments.
\end{remark}
\begin{figure}[thpb]
      \begin{center}

      \includegraphics[scale=0.2]{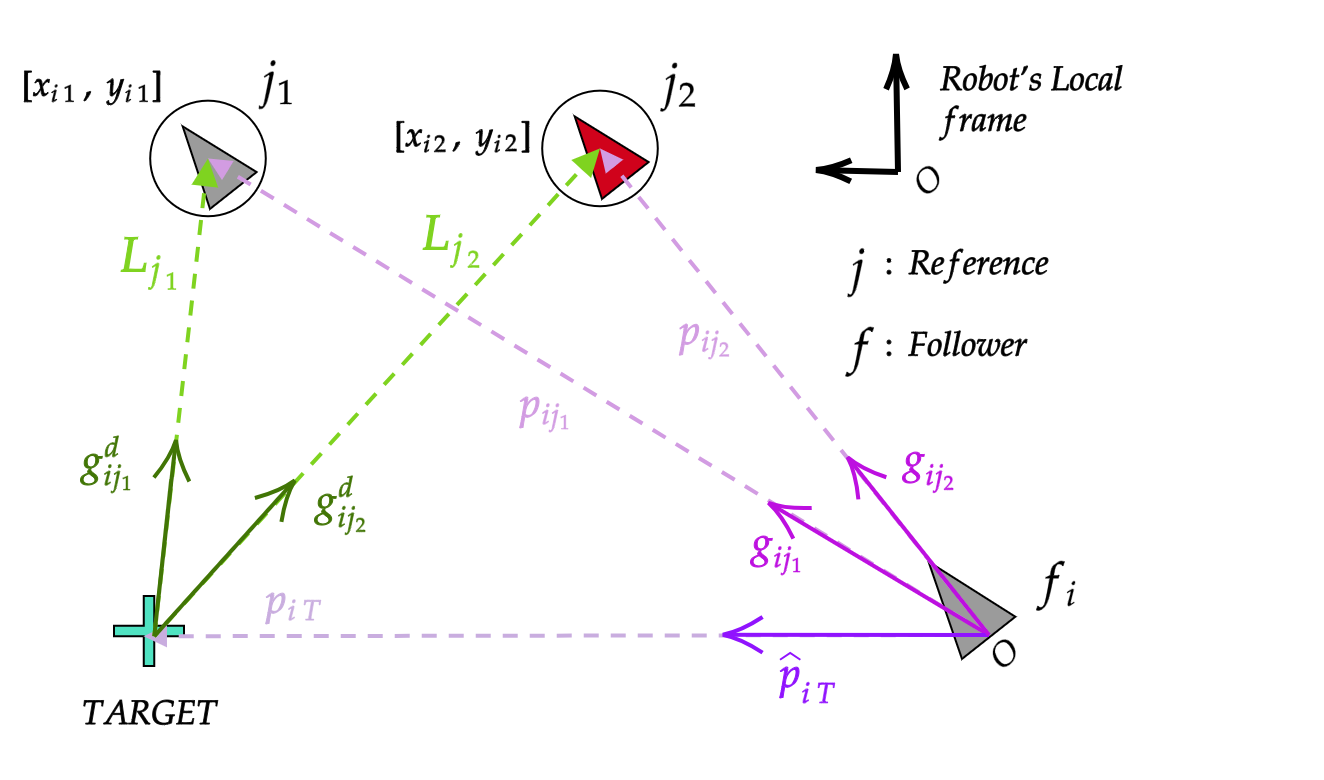}

      \caption{Formation control problem for follower agent $i$.}
      \label{probset}
      \end{center}
\end{figure}
Local references for constructing the formation subproblem (as shown in Fig. \ref{probset}) are introduced next.
\begin{lemma}\label{lem2}
    Given the multi-agent system connected as $\mathcal{G}$ in a leader-follower framework and under Assumption~\ref{A1}, each follower agent $i$ has at least two neighbors in unique positions with non-parallel desired bearings, termed as reference nodes $j_1,j_2 \in \mathcal{N}_i$, such that the removal of any corresponding edge $(i,j_1)\ \text{or}\ (i,j_2) \in \mathcal{E}$, implies that the matrix $\mathcal{B}_{ff}$ is singular (i.e., $det(\mathcal{B}_{ff}) = 0$). 
\end{lemma}
\begin{pf}
Using the ``if part" of Lemma \ref{lem1} under Assumption~\ref{A1} (i.e. $det(\mathcal{B}_{ff})\neq0$), the desired position for any follower agent \( i \), denoted as \( \boldsymbol{p}_i^d \in \mathbb{R}^2, i \in \mathcal{V}_{F} \), is uniquely determined requiring at least two distinct neighbors (which can be a leader or follower) at unique positions with non-parallel desired bearings. If a follower agent \( i \) has at most one neighbor \( j \) with a unique position, or the neighbors at unique positions have parallel desired bearings, then the position of agent \( i \) becomes non-unique, which, by the ``only if" part of Lemma~\ref{lem1}, implies that \( \det(\mathcal{B}_{ff}) = 0 \), which contradicts Assumption \ref{A1}. Therefore, there exist at least two neighbors at unique positions for any $i \in \mathcal{V}_f$, with non-parallel desired bearings for non-singular $\mathcal{B}_{ff}$.\qed
\end{pf} 
Finally, we now formally state the subproblem as follows.
\begin{subproblem}
Consider a formation control problem for a follower agent $i \in \mathcal{V}_{F}$ with respect to its own local frame of reference. Given the relative position of the reference nodes $j_1$ and $j_2$ as $\boldsymbol{p}_{ij_1} = [x_{i1} ,\ y_{i1}]^{\top}$ and $\boldsymbol{p}_{ij_2} = [x_{i2},\ y_{i2}]^{\top}$, respectively. By leveraging Assumption \ref{A2}, we need to achieve \textit{formation tracking}, i.e., the agent $i$ aims to satisfy the formation subtask $\mathcal{F}_i(\boldsymbol{g}(t))$ as
    \begin{align}\label{fsub}
       \mathcal{F}_i(\boldsymbol{g}(t)) :=   \lim_{t \rightarrow \infty}\|\boldsymbol{g}_{ij}(t) - \boldsymbol{g}_{ij}^{d}\| = 0, \forall j \in \{j_1,j_2\}. 
    \end{align}  
\label{sp}
\end{subproblem} 
Next, we show that solving the Subproblem \ref{sp} for each follower agent $i\in\mathcal{V}_F$ also solves the Problem \ref{mainprob}.
\begin{theorem}\label{theo1}
Given a heterogeneous multi-agent system, connected as $\mathcal{G}$ with a leader-follower framework, satisfying Assumptions \ref{A1} and \ref{A2}, and if each follower agent $i$ achieves the formation subtask $\mathcal{F}_i(\boldsymbol{g}(t))$ \eqref{fsub}, then the multi-agent system achieves formation $\mathcal{F}(\boldsymbol{g}(t))$ as in Problem \ref{mainprob}.
\end{theorem}
\begin{pf}
   Given a unique rigid formation $\mathcal{F}(\boldsymbol{g}(t))$, described by the desired bearing configuration $ \boldsymbol{g}^{d}$, formation subtask $\mathcal{F}_i(\boldsymbol{g}(t))$ implies a unique local formation for each follower agent $i$, with respect to their reference nodes, as the local bearings $\boldsymbol{g}_{ij}(t) \rightarrow \boldsymbol{g}_{ij}^{d}, \forall j \in \{j_1,j_2\}$. Therefore, with the leaders moving as predefined and each follower satisfying $\mathcal{F}_i(\boldsymbol{g}(t))$, the collective satisfaction of all formation subtasks, i.e., $\{\mathcal{F}_i(\boldsymbol{g}(t))\}_{{\forall i\in \mathcal{V}_{F}}}$, implies that $\boldsymbol{g}_{ij}(t) \rightarrow \boldsymbol{g}_{ij}^{d}, \forall (i,j) \in \mathcal{E}$, due to unique rigid formation with reference nodes using the Lemma \ref{lem1} and \ref{lem2}, and the desired formation $\mathcal{F}(\boldsymbol{g}(t))$ is achieved. \qed 
\end{pf}

We can observe in Fig. \ref{probset} that the bearing $\boldsymbol{g}_{ij_1}$ ($\boldsymbol{g}_{ij_2}$) of the reference node $j_1$ (reference node $j_2$) with respect to the follower agent $i$, is along the vector $\boldsymbol{p}_{ij_1}$ ($\boldsymbol{p}_{ij_2}$). In order to achieve the desired bearing vectors $\boldsymbol{g}_{ij_1}^{d}$ and $\boldsymbol{g}_{ij_2}^{d}$, the follower agent must reach the bearing-guided target $\boldsymbol{p}_{iT}$ (discussed in Section \ref{Bearing guided Position Det}).
The method for achieving this objective is discussed in the following Section \ref{sec: cbf constraint design and control}.

\section{CBF-inspired Control for Subproblem}\label{sec: cbf constraint design and control}
Here, we set up the CBF-QP-like framework for Subproblem \ref{sp}, where tracking of the desired formation is considered, with the help of the formation subtask $\mathcal{F}_i(\boldsymbol{g}(t))$ \eqref{fsub}, addressed using a bearing-guided target.
\subsection{Determination of Bearing-Guided Target}\label{Bearing guided Position Det} 
Consider the formation subproblem for each follower agent $i$ where we have at least two reference nodes $j_{1},j_{2}$ with their desired bearing vectors \( \boldsymbol{g}_{ij_k}^d = [g_{ij_k}^{d_x},\ g_{ij_k}^{d_y}]^{\top}, k \in \{1, 2\} ,\) which can also be expressed in terms of their slope as \( m_{ij_k}^d = \frac{g_{ij_k}^{d_y}}{g_{ij_k}^{d_x}} \).
The bearing-guided target in the agent $i$'s frame of reference is defined as follows.
\begin{definition}[Bearing-Guided Target \(\boldsymbol{p}_{iT}\)]
Given \\
Subproblem \ref{sp} for agent $i\in \mathcal{V}_{F}$, using desired bearing slopes ($m_{ij_1}^d$ and $m_{ij_2}^d$), desired relative reference node positions ($\boldsymbol{p}_{ij_1}$ and $\boldsymbol{p}_{ij_2}$), and using the desired formation position vector $\boldsymbol{p}_i^d$, for agent $i$, as defined in \eqref{dpos}, the \emph{bearing-guided target} \(\boldsymbol{p}_{iT}\) for follower agent $i$ is defined as $\boldsymbol{p}_{iT} = \boldsymbol{p}_{i}^d - \boldsymbol{p}_i$, and explicitly computed as 
\begin{align}\label{formation_target}
    \boldsymbol{p}_{iT} = A_i^{-1} \boldsymbol{\sigma}_i,
\end{align}
where $A_i = \begin{bmatrix}
        -m_{ij_1}^d & 1 \\
        -m_{ij_2}^d & 1
    \end{bmatrix}, \boldsymbol{\sigma}_i = \begin{bmatrix}
        [-m_{ij_1}^d,\ 1] \boldsymbol{p}_{ij_1} \\
        [-m_{ij_2}^d,\ 1] \boldsymbol{p}_{ij_2} 
    \end{bmatrix}.$
\end{definition}
Furthermore, the bearing-guided target $\boldsymbol{p}_{iT}$ can be interpreted as the unique point of intersection of the two non-parallel lines \( L_{j_1} \) and \( L_{j_2} \), where each line \( L_{j_k} \) passes through the point \(\boldsymbol{p}_{ij_k}\) with slope $m_{ij_k}^d$ (cf. Fig. \ref{probset}). Additionally, at \(\boldsymbol{p}_{iT} = 0\), i.e., $\boldsymbol{p}_i = \boldsymbol{p}_i^d$, the cosine distance \(\left(1 - \langle \boldsymbol{g}_{ij_k}, \boldsymbol{g}^d_{ij_k} \rangle \right)\) is zero, implying alignment with the desired bearing directions due to the unique formation in bearing as defined in Lemma \ref{lem1}. The magnitude of the bearing-guided target vector can be expressed as
$
    \|\boldsymbol{p}_{iT}\| = \|A_i^{-1} \boldsymbol{\sigma}_i\|,
$
and 
it will be leveraged in formulating a CBF-like constraint to enforce the desired formation in the following subsection.

\subsection{Controller Design for Formation Tracking}
We propose a QP-based optimization framework to solve the Subproblem \ref{sp} using the bearing-guided target $\boldsymbol{p}_{iT}$, as in \eqref{formation_target}, for each follower agent $i\in \mathcal{V}_{F}$. The objective is to satisfy the local subtask $\mathcal{F}_i(\boldsymbol{g}(t))$ by aligning the velocity $\boldsymbol{v}_i$ with the prescribed orientation $\boldsymbol{p}_{iT} = \boldsymbol{p}_i^d - \boldsymbol{p}_i$, while regulating the speed $\|\boldsymbol{v}_i\|$ to the leader agent's speed $\|\boldsymbol{v}_{l}\| = \|\dot{\boldsymbol{p}}_{i}^d\|$ (due to Lemma \ref{lem1}), so that the desired bearing configuration $\boldsymbol{g}^d$ is achieved, i.e.,  $\lim_{t \rightarrow \infty}\|\boldsymbol{p}_{iT}\| = 0$.
As the reference nodes $j_1,j_2$ move, the bearing-guided target evolves over time, yielding a time-varying reference trajectory $\boldsymbol{p}_{iT}: \mathbb{R}_{0}^{+} \rightarrow \mathbb{R}^2$. Also under Assumption \ref{A1}, the bearing-guided target ${\boldsymbol{p}}_{iT}$ varies with velocity $\dot{\boldsymbol{p}}_{iT} = \dot{\boldsymbol{p}}_{i}^{d} - \dot{\boldsymbol{p}}_{i}=\boldsymbol{v}_{l} - \boldsymbol{v_i}$. So, we enforce the alignment of velocity $\boldsymbol{v}_i$ along $\boldsymbol{p}_{iT}$ with the help of reach CBF.
For the required alignment between $\boldsymbol{v}_i$ and $\boldsymbol{p}_{iT}$, we define reach CBF $b_i$ based on the angle \(\theta_i := \theta_i(t)\) between \(\boldsymbol{v}_i\) and \(\boldsymbol{p}_{iT}\):
\begin{align}\label{cbfr}
    b_i(\hat{\boldsymbol{v}}_i,\hat{\boldsymbol{p}}_{iT}) := \langle\hat{\boldsymbol{v}}_i,\hat{\boldsymbol{p}}_{iT}\rangle - \cos(\theta_{i\circ}),
\end{align}
where $\hat{\boldsymbol{v}}_i = \frac{{\boldsymbol{v}}_i}{\|{\boldsymbol{v}}_i\|}, \hat{\boldsymbol{p}}_{iT} = \frac{{\boldsymbol{p}}_{iT}}{\|{\boldsymbol{p}}_{iT}\|}$, $\theta_{i\circ} =\cos^{-1}(\langle\hat{\boldsymbol{v}}_i(0),\hat{\boldsymbol{p}}_{iT}(0)\rangle)$, and $\max_{\forall \hat{\boldsymbol{v}}_i, \hat{\boldsymbol{p}}_{iT} \in \mathcal{S}} b_i(\hat{\boldsymbol{v}}_i,\hat{\boldsymbol{p}}_{iT}) = 1 - \cos(\theta_{i\circ})$. The compact set $X_i$ is defined using reach CBF $b_i$ as 
\begin{align}\label{domain}
X_i := \{(\hat{\boldsymbol{v}}_i,\hat{\boldsymbol{p}}_{iT}) \in \mathcal{S} \times \mathcal{S}: b_i(\hat{\boldsymbol{v}}_i,\hat{\boldsymbol{p}}_{iT}) \geq 0\},
\end{align}
and the boundary set $\partial X_i := \{(\hat{\boldsymbol{v}}_i,\hat{\boldsymbol{p}}_{iT}) \in \mathcal{S} \times \mathcal{S}: b_i(\hat{\boldsymbol{v}}_i,\hat{\boldsymbol{p}}_{iT}) = 0\}$. 
{In this work, the notion of safety corresponds to maintaining a task-consistent motion (bounded error with directional alignment), rather than collision avoidance.}
Here for given $\hat{\boldsymbol{p}}_{iT}$, $\forall \hat{\boldsymbol{v}}_i \in \mathcal{S}$, $b_i(\hat{\boldsymbol{v}}_i, \hat{\boldsymbol{p}}_{iT})$ has unique maxima as required by Definition \ref{cbfdef} of reach CBF. Thus, we enforce the reachability of the set $R_i$ for formation tracking, with an arbitrarily chosen small angle $0<\gamma_i \leq \|\theta_{i\circ}\|$ as
\begin{align} \label{reachset}
    R_i:= \{(\hat{\boldsymbol{v}}_i,\hat{\boldsymbol{p}}_{iT}) \in \mathcal{S} \times \mathcal{S}: b_i(\hat{\boldsymbol{v}}_i,\hat{\boldsymbol{p}}_{iT}) \geq \mu_i\},
\end{align}
where $\mu_i = \cos(\gamma_i) - \cos(\theta_{i\circ}) > 0$, which results in $R_i \subset X_i$. Here, the parameter $\gamma_i$ represents the alignment error tolerance between $\hat{\boldsymbol{v}}_i$ and $\hat{\boldsymbol{p}}_{iT}$. 
Now, after simplifying the reach CBF constraint in \eqref{QPdef}, with the reach CBF for formation tracking as defined in \eqref{cbfr}, together with the system dynamics (\ref{sys}), we obtain the linear constraint in $\boldsymbol{u}_i$ as $K_f + K_h^{\top}\boldsymbol{u}_i \geq \frac{\mu_i}{\tau_i}$ with
\begin{align}\label{movet}
&K_h^{\top}(\hat{\boldsymbol{v}}_i,\hat{\boldsymbol{p}}_{iT}) = \hat{\boldsymbol{p}}_{iT}^{\top}\frac{1}{\|\boldsymbol{v}_i\|}\left[I_2 - \boldsymbol{v}_i\boldsymbol{v}_i^{\top}\right]h_{v_i}(\boldsymbol{x}_i), \nonumber \\
&K_f(\hat{\boldsymbol{v}}_i,\hat{\boldsymbol{p}}_{iT}) = \hat{\boldsymbol{p}}_{iT}^{\top}\frac{1}{\|\boldsymbol{v}_i\|}\left[I_2 - \boldsymbol{v}_i\boldsymbol{v}_i^{\top}\right]f_{v_i}(\boldsymbol{x}_i) \nonumber \\ 
&+ \hat{\boldsymbol{v}}_i^{\top}{\frac{1}{\|\boldsymbol{p}_{iT}\|}\left[I_2 - \boldsymbol{p}_{iT}\boldsymbol{p}_{iT}^{\top}\right](\boldsymbol{v}_l - \boldsymbol{v}_i)}. 
\end{align}
Using Theorem \ref{thm1} and the constraint \eqref{movet}, our goal is to reach the set $R_i$ (as in \eqref{reachset}) in time $\tau_i$, where $\hat{\boldsymbol{v}}_i$ and $\hat{\boldsymbol{p}}_{iT}$ are aligned with small angle error $\gamma_i$, i.e., $\cos(\theta_i(\tau_i)) \geq \cos(\gamma_i)$, while being within the compact set $X_i$ (as in \eqref{domain}). 

We now formulate a controller using a QP optimization problem for every agent $i$, using constraint \eqref{movet}, together with a nominal controller $\boldsymbol{u}_{iv}$ for the speed tracking of the moving formation. The feasibility of the proposed QP problem depends on the validity of the reach CBF constraint imposed, which becomes ill-posed for abrupt reversal when the bearing-guided target direction is anti-aligned with the agent’s velocity. Thus, we assume that
\begin{assum} \label{A4}
    The bearing-guided target direction \(\hat{\boldsymbol{p}}_{iT}\) is not anti-aligned with the agent’s velocity \(\boldsymbol{v}_i\), i.e., $\langle\hat{\boldsymbol{p}}_{iT}, \hat{\boldsymbol{v}}_i\rangle \neq -1, \forall t \geq 0$.
\end{assum}
{The above assumption excludes a measure-zero case, and in practice the reach-CBF steers the system away from it, with fail-safe check adding a small perturbation near anti-alignment, which avoids ill-conditioning.}
Furthermore, the above constraint becomes invalid only in the degenerate case when $\left[I_2 - \boldsymbol{v}_i\boldsymbol{v}_i^{\top}\right]h_{v_i} = 0_{2}$ or $\langle\boldsymbol{p}_{iT},\boldsymbol{v}_i^{\perp}\rangle = 0$, where the vector $\boldsymbol{v}_i^{\perp} = P_{\boldsymbol{v}_i}\boldsymbol{v}_i$, $P_{\boldsymbol{v}_i} = \left[I_2 - \boldsymbol{v}_i\boldsymbol{v}_i^{\top}\right]$. However, $\langle\boldsymbol{p}_{iT},\boldsymbol{v}_i^{\perp}\rangle = 0$ implies either alignment, requiring no correction, or anti-alignment, which is excluded by Assumption \ref{A4}. Since the system (\ref{sys}) is controllable ($Null(h_{v_i}) = \emptyset$), the constraint remains valid (i.e. $\left[I_2 - \boldsymbol{v}_i\boldsymbol{v}_i^{\top}\right]h_{v_i} \neq 0_{2}$). After verifying the validity of the reach CBF constraint, the QP-based control formulation is presented below. 
\begin{theorem}
    Consider the reach CBF $b_i(\hat{\boldsymbol{v}}_i,\hat{\boldsymbol{p}}_{iT})$ as in \eqref{cbfr} with constraint \eqref{movet}, where Assumption \ref{A4} holds and if input $\boldsymbol{u}_i$ is obtained using QP optimization problem  
\begin{align}\label{QPws}
    &{\boldsymbol{u}}_i = \arg \min_{\boldsymbol{q} \in \mathbb{R}^2} \quad \frac{1}{2}(\boldsymbol{q} -  \boldsymbol{u}_{iv})^{\top}(\boldsymbol{q}-  \boldsymbol{u}_{iv}), \nonumber \\
    &\text{s.t.} \quad K_{f}(\hat{\boldsymbol{v}}_i,\hat{\boldsymbol{p}}_{iT}) + K_{h}^{\top}(\hat{\boldsymbol{v}}_i,\hat{\boldsymbol{p}}_{iT})\boldsymbol{q} \geq \frac{\mu_i}{\tau_i}, 
\end{align}
with $K_{f}, K_{h}$ as in \eqref{movet}, along with the state-feedback controller
\begin{align}\label{lyac}
    &\boldsymbol{u}_{iv} = h_{v_i}^{-1}(-f_{v_i}  - K_{i1}\boldsymbol{\xi}_{ip} + K_{i1}(\boldsymbol{v}_l - \boldsymbol{v}_i) + K_{i2}\boldsymbol{\xi}_{iv}),
\end{align}
where $K_{i2} = \kappa_{i2}I_{2},K_{i1} = \kappa_{i1}I_{2}, 
\kappa_{i1}>1,\kappa_{i2}>\frac{1+\kappa_{i1}}{2}, \boldsymbol{\xi}_{ip} = \boldsymbol{p}_{iT}, \boldsymbol{\xi}_{iv} = \boldsymbol{v}_l - \boldsymbol{v}_i + K_{i1}\boldsymbol{\xi}_{ip}$,
then the formation subtask $\mathcal{F}_i(\boldsymbol{g}(t))$ as \eqref{fsub} is achieved for an agent $i\in \mathcal{V}_F$.
\label{theo2}
\end{theorem}
\begin{pf}
    We need to prove that for agent $i$, formation $\mathcal{F}_i(\boldsymbol{g}(t))$ is achieved as in \eqref{fsub}. With initial $\boldsymbol{v}_i(0)$, the barrier function $b_i(\boldsymbol{v}_i(0),\boldsymbol{p}_{iT}(0)) = 0$, and the set $R_i$ (as in \eqref{reachset}) is reached as by integrating the constraint in \eqref{movet} over $[0,\tau_i]$ with lower bound $\frac{\mu_i}{\tau_i}$ and $b_i(\hat{\boldsymbol{v}}_i(0),\boldsymbol{p}_{iT}(0))=0$, and following arguments analogous to Theorem \ref{thm1}, it follows that $b_i(\hat{\boldsymbol{v}}_i(\tau_i)) \geq \mu_i$, ensuring $\cos(\theta_i(\tau_i)) \geq \cos(\gamma_i)$.
    {Therefore, the directional CBF constraint ensures alignment independent of the agent's speed as per the nominal controller.}
    In addition to tracking the correct orientation for the moving formation, the agent tracks the speed of moving desired unique position $\boldsymbol{p}_{i}^{d}$, defined in \eqref{dpos}, only if the relative bearing-guided target position $\boldsymbol{\xi}_{ip} \rightarrow 0$ and the relative bearing-guided target velocity $\boldsymbol{\xi}_{iv} \rightarrow 0$ as $ (\boldsymbol{\xi}_{ip},\boldsymbol{\xi}_{iv}) = (0,0) \Rightarrow \boldsymbol{p}_{iT} = \boldsymbol{0}, \boldsymbol{v}_{l} = \boldsymbol{v}_i \Rightarrow \boldsymbol{g}_{ij} = \boldsymbol{g}_{ij}^d,\forall j \in \{j_1,j_2\}$.
    {Next, we control the agent's speed independent of the constraint as below.}
         Consider a Lyapunov function
            $V = \frac{1}{2}\boldsymbol{\xi}_{ip}^{\top}\boldsymbol{\xi}_{ip} + \frac{1}{2}\boldsymbol{\xi}_{iv}^{\top}\boldsymbol{\xi}_{iv}$.
        With the system dynamics \eqref{sys} and the state feedback controller \eqref{lyac}, the nominal error dynamics is written as
        \begin{align}
            &\dot{\boldsymbol{\xi}}_{ip} = \boldsymbol{v}_l - \boldsymbol{v}_i = -K_{i1}\boldsymbol{\xi}_{ip} + \boldsymbol{\xi}_{iv},\label{e1}\\ 
            &\dot{\boldsymbol{\xi}}_{iv} = K_{i1}\dot{\boldsymbol{\xi}}_{ip} -f_{v_i} - h_{v_i}\boldsymbol{u}_i = K_{i1}\boldsymbol{\xi}_{ip} - K_{i2}\boldsymbol{\xi}_{iv}.\label{e2}
        \end{align}
        By differentiating $V$ and substituting \eqref{e1},\eqref{e2}, we get
        $\dot{V} = \boldsymbol{\xi}_{ip}^{\top}\dot{\boldsymbol{\xi}}_{ip} + \boldsymbol{\xi}_{iv}^{\top}\dot{\boldsymbol{\xi}}_{iv}  = \boldsymbol{\xi}_{ip}^{\top}\boldsymbol{\xi}_{iv} - \boldsymbol{\xi}_{ip}^{\top}K_{i1}\boldsymbol{\xi}_{ip} + \boldsymbol{\xi}_{iv}^{\top}K_{i1}\boldsymbol{\xi}_{ip} - \boldsymbol{\xi}_{iv}^{\top}K_{i2}\boldsymbol{\xi}_{iv} \leq  \frac{(1 + \kappa_{i1})}{2}(\|\boldsymbol{\xi}_{ip}\|^{2} + \|\boldsymbol{\xi}_{iv}\|^{2}) - \kappa_{i1}\|\boldsymbol{\xi}_{ip}\|^{2} - \kappa_{i2}\|\boldsymbol{\xi}_{iv}\|^{2} \leq \frac{(1 - \kappa_{i1})}{2}\|\boldsymbol{\xi}_{ip}\|^{2} + \frac{(1 + \kappa_{i1}- 2\kappa_{i2})}{2}\|\boldsymbol{\xi}_{iv}\|^{2}$,
        where the upper bound was obtained using Young's inequality and the control gain $\kappa_{i1}>0$.
         Since $\kappa_{i1}>1,\kappa_{i2}>\frac{1+\kappa_{i1}}{2}$, we have $\dot{V} < 0$ and the equilibrium point $(\boldsymbol{\xi}_{ip},\boldsymbol{\xi}_{iv}) = (0,0)$ is asymptotically stable in $\mathbb{R}^{2}$. Based on the above analysis, it shows that the desired position $\boldsymbol{p}_{i}^{d}$ is tracked and, in turn, $\mathcal{F}_i(\boldsymbol{g}(t))$ is achieved as defined in \eqref{fsub}.\qed
\end{pf}
\begin{remark}
    With Theorem \ref{theo2}, when leaders are in fixed positions $\boldsymbol{p}_{l}^{d}$, stationary formation can be achieved with $\boldsymbol{v}_l = \boldsymbol{0}$ in the constraint \eqref{movet} and the speed tracking controller $\boldsymbol{u}_{iv}$. 
\end{remark}
    Consequently, no separate nominal formation controller is needed.} Using the state feedback controller $\boldsymbol{u}_{iv}$, we rescale the agent's speed $\|\boldsymbol{v}_i\|$ ({not the position errors}) in order to chase the moving bearing-guided target $\boldsymbol{p}_{iT}$, while the desired velocity orientation $\hat{\boldsymbol{v}}_i$ for formation is achieved by the reach CBF constraint. {This separation in the control design for velocity's magnitude and direction decouples the competing objectives. We next conclude for the global formation task $\mathcal{F}(\boldsymbol{g}(t))$ using Theorem \ref{theo1}, \ref{theo2}.  
\begin{theorem}\label{theo4}
Consider the heterogeneous multi-agent system, connected as $\mathcal{G}$ with a leader-follower framework, as defined in Section \ref{sec: prelim and prob defn}, satisfying Assumptions \ref{A1}-\ref{A4}. If each follower agent $i \in \mathcal{V}_F$ applies the control law $\boldsymbol{u}_i(t)$ obtained from the CBF-QP formulation in Theorem~\ref{theo2}, then the multi-agent system achieves the desired formation $\mathcal{F}(\boldsymbol{g}(t))$, thereby solving Problem~\ref{mainprob}.
\end{theorem}

\begin{figure}[th]
     
      \begin{center}
      \vspace{-1em}
      \includegraphics[scale=0.45]{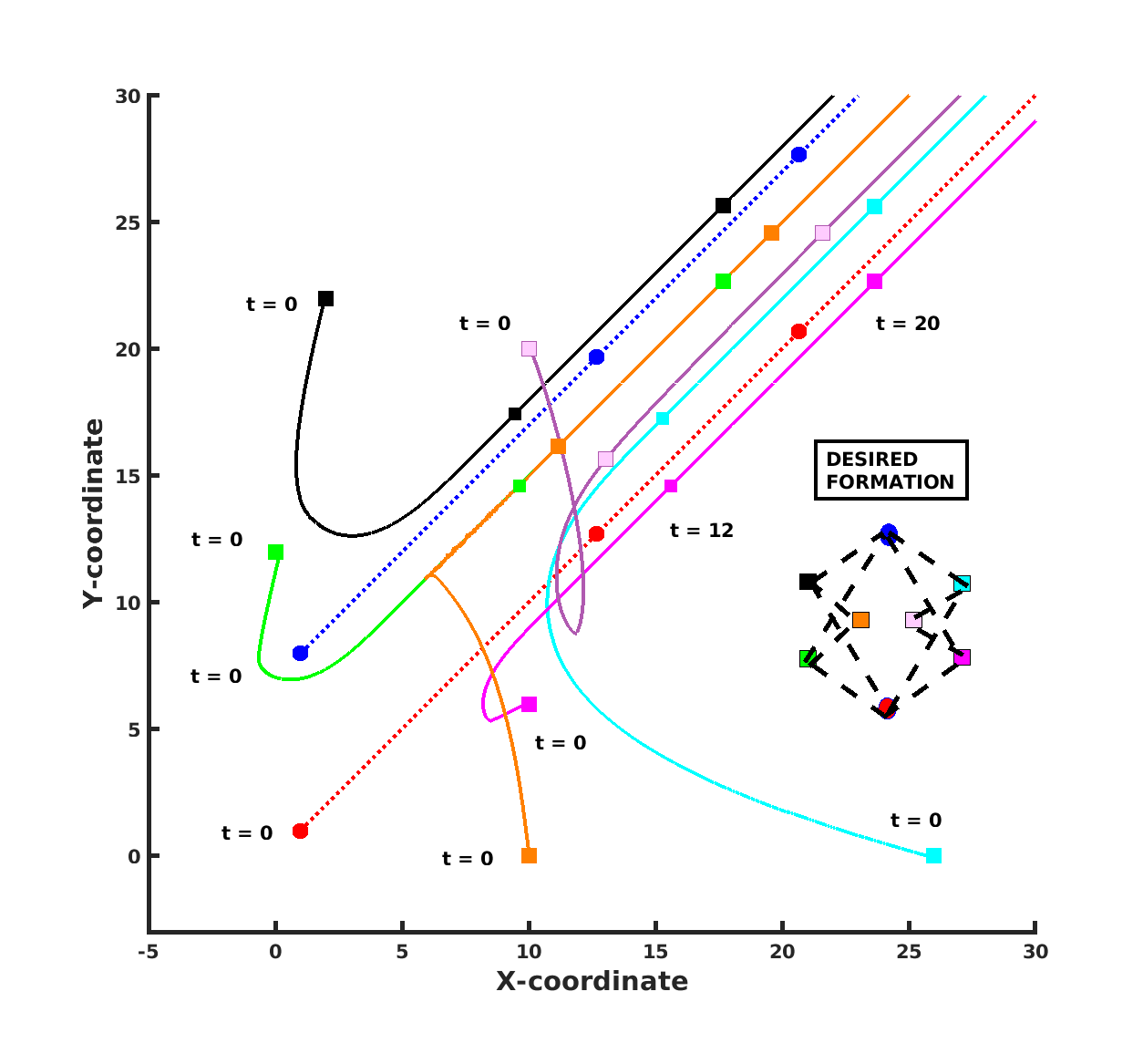}
      \vspace{-1em}
      \caption{Timestamps in s are shown for formation group positions. Formation tracking with two leader agents (red, blue) and six follower agents (green, black, cyan, magenta, orange, pink) is shown. For video visit \url{https://shorturl.at/0E2Rj}.}
      \label{result}
      \end{center}

\end{figure}
\section{Results \& Conclusion}
\label{sec: results}
In this section, we validate our proposed single CBF-like constraint through an integrated CBF-QP framework for formation tracking in a heterogeneous multi-agent system. The simulation includes two leader agents $\{1,2\}$ and six follower agents $\{3,4,5,6,7,8\}$ in a 2D space, connected as $\mathcal{G} = (\{1,2,3,4,5,6,7,8\},\{(3,1),(3,2),(4,1),(4,2),(5,1),\\(5,2),
(6,1),(6,2),(7,3),(7,4),(8,5),(8,6)\})$. The formation task $\mathcal{F}(\boldsymbol{g}(t))$ is defined by the bearing configuration $\boldsymbol{g}^{d} = \{(0.51,0.85),(0.83,-0.55),(0.83, 0.55),(0.51,-0.85),\\(-0.83,0.55),(-0.51,-0.85),(-0.51,0.85),(-0.83,-0.55),\\(-0.70, -0.70),(-0.89, 0.44),(0.89, 0.44),(0.70, -0.70)\}$, in the same order as the edge set $\mathcal{E}$, satisfying Assumption \ref{A1}. The follower agents 3, 6, and 8 are modeled as a double integrator with $f_{v_i}=\boldsymbol{0}, h_{v_i}=I_2$ and agents 4, 5, and 7 are modeled as differential drive (or Ackermann drive with small slip angle) with $f_{v_i} = \boldsymbol{0}, h_{v_i} = \begin{bmatrix}
    \cos(\delta_i) & -\beta_i\sin(\delta_i)\\
    \sin(\delta_i) & \beta_i\cos(\delta_i)
\end{bmatrix}$, where $\delta_i$ is the robot's orientation with $\dot{\delta_i} = [0,\ 1]\boldsymbol{u}_i$ and $\beta_i = \|\boldsymbol{v}_i\|$ {(for system dynamics derivation refer to \cite{sawarkar2026sliding}[II B]).} 
    We consider a case study with leader agents follow predetermined trajectories (shown by red and blue dotted lines in Fig. \ref{result}), while follower agents aim to achieve a desired formation $\mathcal{F}(\boldsymbol{g}(t))$, as defined in Definition \ref{mform}, using CBF-QP optimization framework in Theorem \ref{theo2}. The initial positions for leaders are $\boldsymbol{p}_1 = [1,\ 8]^{\top}$, $\boldsymbol{p}_2 = [1,\ 1]^{\top}$ and for followers are $\boldsymbol{p}_{3} = [0,\ 12]^{\top}$, $\boldsymbol{p}_{4} = [2,\ 22]^{\top}$, $\boldsymbol{p}_{5} = [26,\ 0]^{\top}$, $\boldsymbol{p}_{6} = [10,\ 6]^{\top}$, $\boldsymbol{p}_{7} = [10,\ 0]^{\top}$, $\boldsymbol{p}_{8} = [10,\ 20]^{\top}$. The initial orientations for differential drive followers are $\delta_4 = \frac{\pi}{2}, \delta_5 = -\frac{\pi}{4}, \delta_7 = 1$. The initial velocities are $\boldsymbol{v}_1 = [1,\ 1]^{\top}$, $\boldsymbol{v}_2 = [1,\ 1]^{\top}$, $\boldsymbol{v}_3 = [2,\ 3]^{\top}$, $\boldsymbol{v}_4 = [0,\ 5]^{\top}$, $\boldsymbol{v}_5 = [3.53,\ -3.53]^{\top}$, $\boldsymbol{v}_6 = [1,\ 1]^{\top}$, $\boldsymbol{v}_7 = [0.54,\ 0.84]^{\top}$, $\boldsymbol{v}_8 = [1,\ 1]^{\top}$, and angle tolerance $\gamma_i = 1^{\circ}, \forall i \in \mathcal{V}_{F}$. Using the reach CBF constraint \eqref{movet}, we have validated the formation with the final bearing error after 30 s as $\epsilon_3 = -4.6e^{-5}, \epsilon_4 = -4.6e^{-5}, \epsilon_5 = -4.5e^{-5}, \epsilon_6 = -4.6e^{-5}, \epsilon_7 = 1.87e^{-4}, \epsilon_8 = 1.83e^{-4}$, where $\epsilon_i = |\mathcal{N}_i| - \sum_{j=1}^{|\mathcal{N}_i|}\langle\boldsymbol{g}_{ij},\boldsymbol{g}_{ij}^{d}\rangle$. Formation tracking is achieved with solid lines green, black, cyan, magenta, orange, and pink as the path taken by follower agents 3-8, respectively (cf. Fig. \ref{result}).

In conclusion, this paper has proposed a control strategy for formation tracking, formulated through a single CBF-like constraint embedded in a CBF-QP-like optimization scheme for heterogeneous multi-agent systems. Future work aims to extend our approach for formation tracking with aerial vehicles in 3D space, navigating cluttered environments with accelerating obstacles.

\bibliographystyle{unsrtnat}
\bibliography{references}  






\end{document}